\documentclass{cccg24}
\usepackage{graphicx,amssymb,amsmath}
\usepackage{lineno}
\usepackage{xspace}
\usepackage{url}
\usepackage{hyperref}

\newtheorem{definition}{Definition}
\newtheorem{observation}{Observation}
\newcommand{\match}{{\sc Matching}\xspace}

\usepackage[ruled,vlined,linesnumbered,titlenumbered]{algorithm2e}
\DontPrintSemicolon
\SetKwInput{Input}{Input}
\SetKwInput{Output}{Result}
\SetKw{KwBreak}{break}
\SetKw{Break}{break}
\SetKw{Continue}{continue}
\SetKwRepeat{Do}{do}{while}
\usepackage[disable]{todonotes}

\title{Finding maximum matchings in RDV graphs efficiently}

\author{Therese Biedl\thanks{David R. Cheriton School of Computer Science, University of
Waterloo, {\tt biedl@uwaterloo.ca}}
	\and
	Prashant Gokhale \thanks{David R. Cheriton School of Computer Science, University of
Waterloo, {\tt  prashant.gokhale@uwaterloo.ca}}}

\index{Author, First}
\index{Researcher, Second}

\begin{document}
\thispagestyle{empty}
\maketitle

\begin{abstract}
In this paper, we study the maximum matching problem in \emph{RDV graphs}, i.e., graphs that are \textbf{v}ertex-intersection graphs of \textbf{d}ownward paths in a \textbf{r}ooted tree.   We show that this problem can be reduced to a problem of testing (repeatedly) whether a vertical segment intersects one of a dynamically changing set of horizontal segments, which in turn reduces to an orthogonal ray shooting query.   Using a suitable data structure, we can therefore find a maximum matching in $O(n\log n)$ time (presuming a linear-sized representation of the graph is given), i.e., without even looking at all edges.  
\end{abstract}

\section{Introduction}

The \match~problem is one of the oldest problems in the history of graph theory and graph algorithms:    Given a graph $G=(V,E)$, find a \emph{matching} (a set of pairwise non-adjacent edges) that is \emph{maximum} (has the largest possible number of edges).      See for example extensive reviews of the older history of matchings and its applications in \cite{Berge73, LP86}.
The fastest known algorithm for general graphs runs in $O(\sqrt{n}m)$ time (\cite{MikaliVazirani1980}, see also \cite{Vaz20}).
There have been some recent break-throughs for algorithms for maximum flow, culminating in an algorithm with almost-linear run-time $O(m^{1 + o(1)})$ \cite{CKLPGS23}; this immediately implies an almost-linear algorithm for \match in bipartite graphs. See also \cite{ChuzhoyK23} for a purely combinatorial almost-linear algorithm for the same problem.

\paragraph{Greedy-algorithm and interval graphs.}
Naturally one wonders whether truly linear-time algorithms (i.e., with $O(m+n)$ run-time) exist, at least if the graphs have special properties.   One natural approach for this is to use the greedy-algorithm for \match shown in Algorithm~\ref{alg:greedy_basic}, which clearly takes linear time. With a suitable vertex order this will always find the maximum matching (enumerate the vertices so that matched ones appear consecutively at the beginning); the challenge is hence to find a vertex order (without knowing the maximum matching) for which the greedy-algorithm is guaranteed to work.

\begin{algorithm}[ht]
\caption{Greedy-algorithm for matching}
\label{alg:greedy_basic}
\Input{A graph $G$ with a vertex order $v_1,\dots,v_n$}
Initialize the matching $M=\emptyset$\;
\For{$i=1,\dots,n$}{
	\If{$v_i$ is not yet matched and has unmatched neighbours}{
		among all unmatched 
		neighbours of $v_i$, let $v_j$ be the one that minimizes $j$ \;
		\tcp{$j>i$, for otherwise $v_j$ would have been matched earlier}
		add $(v_i,v_j)$ to matching $M$\;
	}
}
\KwRet{$M$}
\end{algorithm}

One graph class where this can be done is the \emph{interval graphs}, i.e., the
intersection graphs of horizontal segments in the plane.   
It was shown by Moitra and Johnson \cite{MoitraJ89}  that the
greedy-algorithm always finds a maximum matching in an interval graph as
long as we sort the vertices by left endpoint of their intervals.   
This gives an $O(m+n)$ algorithm for interval graphs since an interval representation can be found in $O(m+n)$ time \cite{BL76}.  Liang and Rhee \cite{liang1993} improve this further (presuming an interval representation is given) by rephrasing
the greedy-algorithm as follows (see also Algorithm~\ref{alg:greedy_delay}).   
Rather than adding an edge to the matching when the left endpoint (i.e., the one with the smaller index) is encountered, we add it when the right endpoint is encountered.    We also explicitly maintain a data structure $F$ that stores the \emph{free} vertices, by which we mean vertices that were processed already but are as-of-yet unmatched.       Liang and Rhee \cite{liang1993} implement
$F$ with a balanced binary search tree (storing left endpoints of intervals);
then all operations required on $F$ can be performed in $O(\log n)$ time.
This therefore leads to an $O(n\log n)$ time algorithm for solving \match
in interval graphs; in particular this is \emph{sub-linear} run-time if the
graph has $\omega(n\log n)$ edges. This runtime can be easily improved to $O(n \log{\log{n}})$ by using a van Emde Boas tree \cite{vEB}, as observed by Liang and Rhee \cite{Rhee1995}.

\begin{algorithm}[ht]
\caption{Delayed-Greedy-algorithm}
\label{alg:greedy_delay}
\Input{A graph $G$ with a vertex order $v_1,\dots,v_n$}
Initialize the matching $M=\emptyset$\;
Initialize the set of free vertices $F=\emptyset$\;
\For{$j=1,\dots,n$}{
	\eIf{$v_j$ has neighbours in $F$}{
		among such neighbours, let $v_i$ be the one that minimizes $i$\;
		add $(v_i,v_j)$ to matching $M$\;
		delete $v_i$ from $F$
	}
	{
		add $v_j$ to $F$
	}
}
\KwRet{$M$}
\end{algorithm}

\paragraph{Our results.}
In this paper, we take inspiration from \cite{liang1993} and develop sub-linear algorithms for \match in \emph{RDV graphs}, i.e., graphs that can be represented as vertex-intersection graphs of downward paths in a rooted tree $T$. (This is called an \emph{RDV representation}; formal definitions will be given below.)    
RDV graphs were introduced by Gavril \cite{Gavril75}; many properties
have been discovered and for many problems efficient algorithms have been
found for RDV graphs \cite{BabelPT96, LinS14, LinT15, PT24}, quite frequently in contrast to only slightly bigger graph classes where the problem turns out to be hard. It is easy to see that all interval graphs are RDV graphs, so our results re-prove the results for interval graphs from \cite{liang1993}.

RDV graphs can be recognized in polynomial time, and along the way 
an RDV representation is produced \cite{Gavril75}.
(The run-time has been improved, and even a linear-time algorithm has been claimed but without published details;    see \cite[Section 2.1.4]{ChaplickThesis} for more on the history.)

We show in this paper that if we are given an $n$-vertex
graph $G$ with an RDV representation on a tree $T$, then we can  
find a maximum matching in $O(|T|+n\log n)$ time.   There always
exists an RDV representation of $G$ with $|T|\in O(n)$, so if
we are given a suitable one then the run-time becomes $O(n\log n)$,
hence sub-linear.    

Our idea is to use the delayed-greedy-algorithm (Algorithm~\ref{alg:greedy_delay}), and to pick a suitable data structure for the set $F$ of free vertices.    
The key ingredient here is that `does $v_j$ have a neighbour in $F$' can
be re-phrased, using the RDV representation, as the 
question whether a vertical segment intersects an element of a dynamically changing set of 
horizontal line segments, and if so, to return the one with maximal 
$y$-coordinate among them.    This question in turn can be phrased as an orthogonal
ray-shooting query, for which suitable data structures are known to exist.
The current best implementation of them uses linear space and $O(\log n)$ time per operation
\cite{GiyoraK09}; this gives our result since we need $O(n)$ operations.
We use the ray-shooting data structure as a black box, so if the run-time
were improved (e.g.~one could dream of
$O(\log \log n)$ run-time if coordinates are integers in $O(n)$, as
they are in our application)
then the run-time of our matching-algorithm would likewise improve. Finally, we also study some possible improvements and extensions. 

\paragraph{Other related results: }   There are a number of other results concerning fast algorithms to solve \match in intersection graphs of some geometric objects.   The results for interval graphs were extended to \emph{circular arc graphs} (intersection graphs of arcs of a circle) \cite{liang1993}.    In an entirely different approach, \match can also be solved very efficiently in \emph{permutation graphs} (intersection graphs of line segments connecting two parallel lines) \cite{Rhee1995}; see also \cite{BKKMSU} for an (unpublished) matching-algorithm for permutation graphs that is slower but beautifully uses range queries to find the matching.
We should note that RDV graphs are unrelated to circular arc graphs and permutation graphs (i.e., neither a subclass nor a superclass); Figure~\ref{fig:RDV_example} gives a specific example.   As such, these results do not directly impact ours or vice versa.

Permutation graphs are a special case of \emph{co-comparability graphs}, i.e., graphs for which the complement has an acyclic transitive orientation; these can also be viewed as intersections of curves between two parallel lines \cite{GRU83}.   
For these, maximum matchings can be found in linear time \cite{Mertzios2018}.   

Finally, the greedy-algorithm actually works beyond interval graphs;
in particular Dahlhaus and Karpinski \cite{DAHLHAUS199879}
showed that it finds the maximum matching
for \emph{strongly chordal graphs}.
These are the graphs that are \emph{chordal} (every cycle $C$ of length at least 4 has a
\emph{chord}, i.e., an edge between two non-consecutive vertices of $C$),
and where additionally every even-length cycle $C$ has a chord $(v,w)$ such that
an odd number of edges of $C$ lie between $v$ and $w$.
The question whether chordal graphs have a linear-time
algorithm for \match remains open. But likely the answer is no, because
as argued in \cite{DAHLHAUS199879}, a linear-time algorithm for testing the existence of a \emph{perfect matching}
(i.e., a matching of size $n/2$) in a chordal graph
would imply a linear-time algorithm for the same problem in any bipartite graph that
is \emph{dense} (has $\Theta(n^2)$ edges).

\begin{table}[ht]
\begin{tabular}{@{}c|c|c@{}}
 \textbf{Graph Class} & \textbf{Runtime} & \textbf{Reference} \\ 
\hline
 Interval graphs & $O(n \log{n})$  & \cite{liang1993} \\  
 Interval graphs & $O(n \log\log n)$  & Section~\ref{sec:interval} \\  
 Circular arc graphs &  $O(n \log{n})$ & \cite{liang1993}\\
 Permutation graphs & $O(n \log{\log{n}})$ & \cite{Rhee1995} \\
 Strongly chordal graphs & $O(n + m)$ & \cite{DAHLHAUS199879} \\
 Co-comparability graphs & $O(n + m)$ & \cite{Mertzios2018}\\
 RDV graphs & $O(n \log{n})$ & Section~\ref{sec:RDV} \\
\end{tabular}
\caption{Existing and new results for \match in some classes of
graphs, presuming a suitable intersection representation is given and sufficiently small.}
\label{ta:overview}
\end{table}

Table~\ref{ta:overview} gives an overview of existing and new results
for \match in some classes of intersection graphs of objects.    
Our paper is structured as follows.    After reviewing some background in
Section~\ref{sec:background}, 
we give our main result for RDV graphs in Section~\ref{sec:RDV}.
We briefly discuss interval graphs, as well as other possible
extensions and open problems
in Section~\ref{sec:AdditionalResults}.

\section{Background}
\label{sec:background}

In this paper  we study vertex-intersection graphs of subtrees of trees.   We first define this formally, and then restrict the attention to a specific subclass.

\begin{definition} Let $G$ be a graph.   A \emph{representation of $G$ as a vertex-intersection graph of subtrees of a tree} consists of a \emph{host-tree} $T$ and, for each vertex $v$ of $G$, a subtree $T(v)$ of $T$ such that $(v,w)$ is an edge of $G$ if and only if $T(v)$ and $T(w)$ share at least one node of $T$.
\end{definition}

Such an intersection representation is sometimes also called a \emph{clique-tree}, and slightly abusing notation, we use the word clique-tree also for the host tree $T$ where convenient.
As convention, we use the term `node' for the vertices of the clique-tree, to distinguish them from the vertices of the graph represented by it.
It is well-known that a graph has a clique-tree if and only if it is chordal \cite{Gavril74}.   We now review some properties of clique-trees that have been rooted.

\begin{definition}
    Let $G$ be a graph with a rooted clique-tree $T$.    For any vertex $v$, let $t(v)$ be the topmost (closest to the root) vertex in the subtree $T(v)$ of $v$.
    A \emph{bottom-up enumeration}  of $G$ is a vertex order obtained by sorting vertices by decreasing distance of $t(v)$ to the root, breaking ties arbitrarily.
\end{definition}

Note that this bottom-up enumeration can be computed in $O(|T|+n)$ time, presuming every vertex $v$ stores a reference to $t(v)$.

It will be convenient to assign points in $\mathbb{R}^2$ to the nodes of clique-tree $T$ as follows.   First, fix an arbitrary order of children at each node, and then enumerate the leaves of $T$ as $L_1,\dots,L_\ell$ from left to right.    For every node $i$ in $T$, let $\ell(i)$ be the leftmost (i.e., lowest-indexed) leaf that is a descendant of $i$, and set $x(i)$ to be the index of $\ell(i)$.    We also need the notation $r(i)$ for the rightmost leaf that is a descendant of $i$.
Also, define $y(i)$ to be the distance of node $i$ from the root of the clique-tree.   Figure~\ref{fig:RDV_example} shows each node $i$ drawn at point $(x(i),y(i))$ (where $y$-coordinates increase top-to-bottom).    We can compute $x(\cdot)$ with a post-order traversal and $y(\cdot)$ with a BFS-traversal of host-tree $T$ in $O(|T|)$ time.

\begin{figure}[ht] 
\hspace*{\fill}\includegraphics[width=0.6\linewidth,page=3]{RDV_notInterval.pdf} \\
\includegraphics[width=0.99\linewidth,page=1]{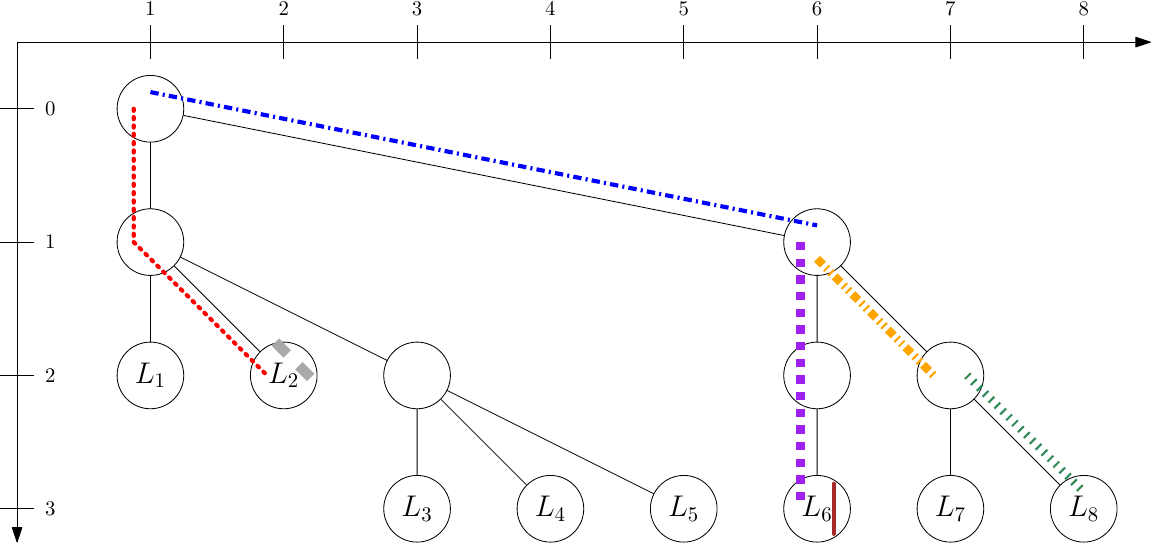}
\caption{An RDV graph together with one possible RDV representation (for illustrative purposes the clique-tree is much bigger than needed).    Nodes are drawn at their coordinates, and vertices are enumerated in bottom-up order.   The graph is neither a circular arc graph nor a permutation graph.}
\label{fig:RDV_example}
\end{figure}

\paragraph{RDV graphs and friends:}
Numerous subclasses of chordal graphs can be defined by studying graphs that have a clique-trees with particular properties.   
Most prominent here is the idea to require that $T(v)$ is a path.    This gives the \emph{path graphs} (also known as \emph{VPT graphs}).
One can further restrict the paths to be directed (after imposing some edge-directions onto the clique-tree); these are the
\emph{directed path graphs}.    One can restrict this even further by requiring that the directions of the clique-tree are
obtained by rooting the clique-tree, and this is the graph class that we study.

\begin{definition} A \emph{rooted directed path graph} (or \emph{RDV graph} \cite{MonmaW86})
is a graph that has an \emph{RDV representation}, i.e., a clique-tree that has been rooted and for every vertex $v$ the subtree $T(v)$ is a \emph{downward path}, i.e., a path that begins at some node and then always goes downwards.
\end{definition}

See Figure~\ref{fig:RDV_example} for an example of an RDV representation.%
\footnote{`RDV' comes from `rooted directed vertex-intersection'.   Gavril called these `directed path graphs' \cite{Gavril75}, but this later got used for the more general graphs where the directions need not be obtained via rooting.}
\section{Matching in RDV graphs}
\label{sec:RDV}

Assume for the rest of this section that we are given an RDV representation of a graph $G$.
In what follows, we will often use `$P(v)$' in place of `$T(v)$' for the subtree of a vertex $v$, to help us remind ourselves that these are downward paths rather than arbitrary trees.
Recall that $t(v)$ denotes the top (closest to the root) node of $P(v)$; because we have a downward path (rather than an arbitrary tree) representing $v$ we can now also define $b(v)$ to be the bottom node of $P(v)$.   

For run-time purposes we presume that `the RDV representation is given' means that we have a rooted tree $T$ and (for each vertex $v$ of $G$) two references $b(v)$ and $t(v)$ to the nodes of $T$ that define the downward path.

Farber \cite{FarberThesis} showed that RDV graphs are strongly chordal,  and
Dahlhaus and Karpinski \cite{DAHLHAUS199879} 
showed that the greedy algorithm works correctly on strongly chordal graphs if we consider vertices in a so-called \emph{strong elimination order} (which is usually assumed to be given with a strongly chordal graph).   This suggests that the
greedy-algorithm works for RDV graphs, but there is one missing piece: How do
we get  a strong elimination order from an RDV representation efficiently?
This is very easy (use the bottom-up enumeration), and the proof that it works
is not hard, but requires some more definitions and is therefore delayed to
the appendix.

\begin{theorem}
\label{thm:RDVstronglyChordal}
    Let $G$ be a graph with a given RDV representation. Then the greedy matching algorithm, applied to a bottom-up enumeration, returns a maximum matching. 
\end{theorem}

Exactly as in \cite{liang1993}, to achieve a sub-linear  run-time we will not use the greedy-algorithm directly but instead use the equivalent delayed-greedy-algorithm (Algorithm~\ref{alg:greedy_delay}).
The main bottleneck for the run-time is then to implement a data structure for the set $F$ of free vertices. Such a data structure should store indexed vertices and must support 
the following three operations:
\begin{itemize}
    \item [A.] insert a new vertex
    \item [B.] delete a vertex
    \item [C.] query for the smallest neighbour, i.e., given a vertex $v_j$ not in $F$, either determine that $v_j$ has no neighbours in $F$, or 
return the neighbour $v_i$ of $v_j$ in $F$ that minimizes index $i$.
\end{itemize}


The first two operations are straightforward, but the third one is non-trivial if we want to use $o(\mathit{degree}(v_j))$ time.
To this end, we reduce adjacency queries in an RDV graph to the question of whether a horizontal segment intersects a vertical segment.   We need some definitions first.

\begin{definition}
Let $G$ be a graph with an RDV representation. For each vertex $v$, define the following (see Figure~\ref{fig:RDV_horsegment} for examples): 
\begin{itemize}
\item The \emph{horizontal segment} $s(v)$ of $v$ is the segment
between the point of $t(v)$ and $\left( x(r(t(v))) , y(t(v)) \right)$, i.e., it
extends rightward until it is above the rightmost descendant of $t(v)$. 
\item The \emph{vertical segment} 
$q(v)$ of $v$ is the segment between the point of $b(v)$ and $(x(b(v)), y(t(v))$, i.e., it extends upward 
until it is to the right of $t(v)$. 
\end{itemize}
\end{definition}

Recall that $t(v)$ has the same $x$-coordinate as its leftmost descendant, so the $x$-range of segment $s(v)$ is exactly the range of $x$-coordinates among descendants of $t(v)$.    We also note that the name `$q$' for the vertical segment was chosen since this will be used to implement the query operation.

\begin{figure}[ht]
\centering
\includegraphics[width=0.99\linewidth,page=2]{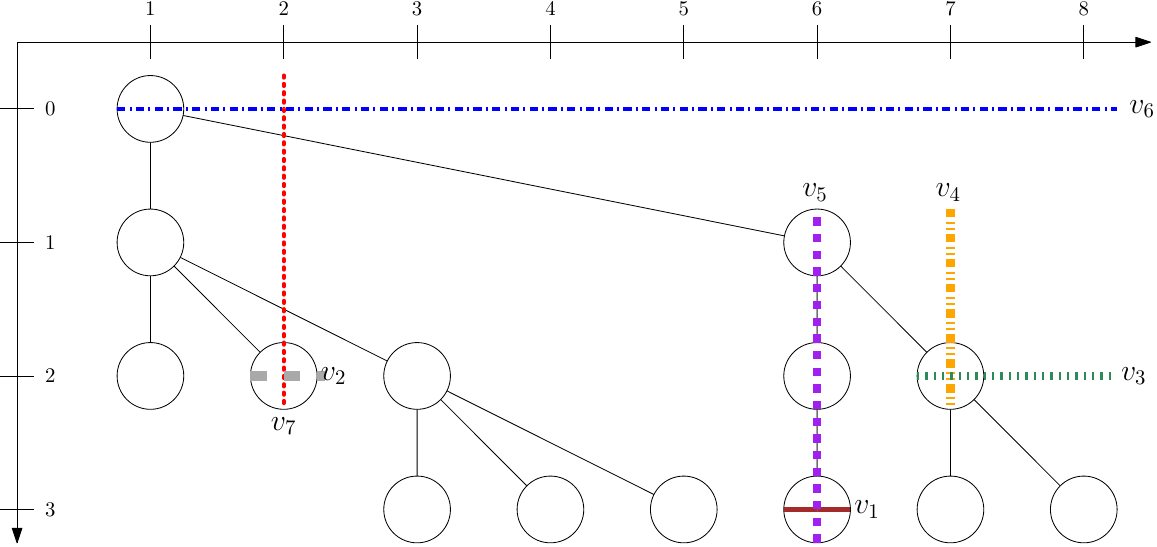}
\caption{Mapping the vertices of the example in Figure~\ref{fig:RDV_example} to horizontal and vertical segments (not all are shown).}
\label{fig:RDV_horsegment}
\label{fig:RDV_versegment}
\end{figure}

 
\begin{theorem}
\label{thm:RDVgeometric}
Let $G$ be a graph with an RDV representation and let $v_1,\dots,v_n$ be the bottom-up enumeration of vertices.     Then for any $i<j$,
edge $(v_i,v_j)$ exists if and only if the vertical segment $q(v_j)$ intersects the horizontal segment $s(v_i)$.
\end{theorem}
\begin{proof}
Since $q(v_j)$ is a vertical segment and $s(v_i)$ is a horizontal segment, they intersect if and only if both the $x$-coordinates
and $y$-coordinates line up correctly, i.e., $x(t(v_i)){=}x(\ell(t(v_i)))\leq x(b(v_j)) \leq x(r(t(v_i)))$ and $y(t(v_j)) \leq y(t(v_i)) \leq y(b(v_j))$.   


%
%

Assume first that edge $(v_i,v_j)$ exists, which means that $P(v_i)$ and $P(v_j)$ have a node $u$ in common.   Among all such nodes $u$, pick the one that is closest to the root; this implies $u\in \{t(v_i),t(v_j)\}$.   By $i<j$ we actually know $u=t(v_i)$, because if $u\neq t(v_i)$ then $u{=}t(v_j)$ would be a strict descendant of $t(v_i)$ and have larger $y$-coordinate, contradicting the bottom-up elimination ordering.
Since $u\in P(v_j)$, node
$b(v_j)$ is a descendant of $u$ (which in turn is a descendant of $t(v_j)$), so $y(t(v_j))\leq y(u){=}y(t(v_i)) \leq y(b(v_j))$ and the $y$-coordinates line up. 
The $x$-coordinates line up since $b(v_j)$ is a descendant of $t(v_i){=}u$
and the horizontal segment $s(v_i)$ covers all such descendants.

Assume now that the segments intersect.
By $y(t(v_j))\leq y(t(v_i))\leq y(b(v_j))$ then path $P(v_j)$ contains a node (call it $u$) with $y(u)=y(t(v_i))$.  If $u$ equals $t(v_i)$ then $P(v_i)$ and $P(v_j)$ have node $t(v_i)$ in common and $(v_i,v_j)$ is an edge as desired.   If $u\neq t(v_i)$, then these two nodes (with the same $y$-coordinate) have a disjoint set of descendants, so the intervals $I_u=[x(\ell(u)),x(r(u))]$ and $I_i=[x(\ell(t(v_i))),x(r(t(v_i)))]$ are disjoint.  Since $b(v_j)$ is a descendant of $u\in P_j$, we have $x(b(v_j))\in I_u$, but since the $x$-coordinates line up we have $x(b(v_j))\in I_i$.    This is impossible.
\todo[color=green]{FYI: I rewrote the proof a bit, the previous version assumed $x(r(t(v_i))<x(\ell(u)$ but we dont know that.}
\end{proof}

In light of this insight, we now can reformulate our requirements on a data structure for $F$ as follows.
We want to store horizontal segments (associated with vertices of a graph) and must be able to support 
the following three operations:
\begin{itemize}
    \item [A'.] insert a new horizontal segment,
    \item [B'.] delete a horizontal segment,
    \item [C'.] query whether a vertical segment $q(v_j)$ intersects a segment in $F$, and if so, among all intersected segments return the segment $s(v_i)$ that maximizes the $y$-coordinate.
\end{itemize}

We can reformulate C' as a ray-shooting query as follows.   Replace the vertical segment $q(v_j)$ by a vertical ray $\vec{q}(v_j)$ obtained by directing $q(v_j)$ upward.   (So the ray originates at the point of $b(v)$ and goes vertically towards smaller $y$-coordinates.)

\begin{observation}
To perform operation $C'$, it suffices to do the following:
\begin{itemize}
    \item [C''.] 
   query whether a ray $\vec{q}(v_j)$ intersects a segment in $F$, and if so, among all intersected segments return the first segment $s(v_i)$ that is hit by the ray.
\end{itemize}
\end{observation}
\begin{proof}
At the time of the query, the set $F$ of free vertices contains only segments of vertices $v_i$ with $i<j$.
Therefore all segments intersected by ray $\vec{q}(v_j)$ have $y$-coordinate at least $y(t(v_j))$, and also intersect the segment $q(v_j)$.   So we will only report correct segments. Since the ray is vertically upward (while $y$-coordinates increase in downward direction), the first segment that is hit is the one that maximizes the $y$-coordinate.
\end{proof}

Operation C'' 
is the well-known \emph{orthogonal ray-shooting} problem, and operations A' and B' means that we want a dynamic variant. Many data structures have been developed for this (some for more general versions), see for
example \cite{1990MehlhornNaher}, \cite{2003KaplanMoladTarjan} for older results with slower processing time.
For the orthogonal ray shooting problem specifically, the best run-time bounds achieved are by Giyora and Kaplan \cite{GiyoraK09}, who
showed how to implement all three operations in $O(\log n)$ time, using $O(n)$ space (assuming the data structure stores up to $n$ items).
Later on this was generalized to drop the requirement of orthogonality \cite{2021Nekrich} without affecting space or runtime.
Some of these data structures assume that the line segments are disjoint. The horizontal segments we have defined earlier are not necessarily
disjoint, but we can make them disjoint (without affecting the outcome) by adding $\tfrac{n-i}{n}$ to the $y$-coordinate of $s(v_i)$. 
\todo{FYI: need to make segments disjoint!   For future rewrites (e.g. thesis) we should probably do this right at the definition
already, and then adjust Theorem 2's proof.}

With this, we can put everything together into our main theorem.

\begin{theorem}
\label{thm:RDVmain}
Given an $n$-vertex graph $G$ with an RDV representation $T$, the maximum matching of $G$ can be found in $O(|T|+n\log n)$ time.
\end{theorem}
\begin{proof}
Parse $T$ to compute the $x$-coordinates and $y$-coordinates of all nodes in $T$, then 
bucket-sort the vertices by decreasing $y(t(v))$ to obtain the
bottom-up order $v_1,\dots,v_n$ in $O(|T|+n)$ time.   By Theorem~\ref{thm:RDVstronglyChordal} 
applying the greedy-algorithm with this vertex-ordering will give a maximum matching.    Using the
delayed greedy-algorithm, the run-time of the algorithm is reduced to performing operations A-C 
$O(n)$ times.    By storing the free set $F$ as horizontal segments, this by Theorem~\ref{thm:RDVgeometric}
is the same as performing operations A', B' and C'' $O(n)$ times. Using a suitable data structure for orthogonal ray shooting \cite{GiyoraK09}, this takes $O(\log n)$ time per operation and hence $O(n\log n)$ time in total.
\end{proof}

One can easily argue that any RDV graph has an RDV representation $T$ with $|T|\in O(n)$,
for otherwise two adjacent nodes of $T$ are used by the same set of subtrees and could be combined into one.
So the run-time becomes $O(n\log n)$
if a suitably small RDV representation is given. 

\medskip
Recall that for interval graphs, an improvement of the run-time for matching from $O(n\log n)$ 
to $O(n\log\log n)$ is possible by exploiting that all intervals can be described via integers
in $O(n)$ and storing $F$ using van Emde Boas trees \cite{Rhee1995}.
This naturally raises an open question:    Could the run-time of Theorem~\ref{thm:RDVmain} also be improved to $O(n\log\log n)$ time, 
presuming $|T|\in O(n)$?   The bottleneck for this would be to improve the run-time for
the orthogonal ray-shooting data structure if all coordinates are (small) integers.  
This question was explicitly asked by Giyora and Kaplan \cite{GiyoraK09}, and appears to be still open.    Could we at least achieve run-time $O(n (\log\log n)^k)$ for some constant $k$ for RDV graphs?

\section{Clique trees where subtrees have few leaves}
\label{sec:AdditionalResults}

\label{sec:fewLeaves}

An RDV graph is a chordal graph with a rooted clique-tree where every subtree $T(v)$ has exactly one leaf.   
A natural generalization of this graph class are the chordal graphs with a rooted clique-tree where
every subtree $T(v)$ has at most $\Delta$ leaves.   (A very similar concept was introduced by
Chaplick and Stacho under the name of \emph{vertex leafage} \cite{ChaplickS14}; the only difference is that they considered 
unrooted clique trees and so count the root of $T(v)$ as leaf if it has degree 1.)

\begin{theorem}
Let $G$ be a graph with a rooted clique-tree $T$ where all subtrees have at most $\Delta$ leaves.    
Then the greedy-algorithm applied to the bottom-up enumeration can be implemented in $O(|T|+\Delta n\log n)$ time.
\end{theorem}
\begin{proof}
For each vertex $v$, split $T(v)$ into $k\leq \Delta$ paths $P_1(v),\dots,P_k(v)$, each connecting the root $t(v)$ of $T(v)$ to a leaf, such that their union covers all of $T(v)$.    Define segment $s(v)$ as before (it only depends on $t(v)$), and define $k$ query-segments $q_1(v),\dots,q_k(v)$ for the paths.   One easily verifies that for $i<j$ vertex $v_j$ is a neighbour of $v_i$ if and only if at least one of $q_1(v_j),\dots,q_k(v_j)$ intersects $s(v_i)$.    So to perform operation C, we do a ray-shooting query for each of $\vec{q_1}(v_j),\dots,\vec{q_k}(v_j)$ and
choose among the returned segments (if any) the one that has maximum $y$-coordinate.  With this operation C can be implemented in 
$O(\Delta \log n)$ time.    All other aspects of the greedy-algorithm are exactly as in Section~\ref{sec:RDV}.
\end{proof}

Unfortunately, this does not improve the time to find maximum matchings for such graphs, because there is no guarantee that the greedy-algorithm finds a maximum matching when applied with a bottom-up enumeration.
To see a specific example, consider the \emph{directed path graphs} (recall that these are obtained by requiring $T(v)$ to be a directed path after directing the clique-tree, but the edge-directions need not come from rooting the clique-tree).
This is a strict superclass of RDV graphs, for example the graph in Figure~\ref{fig:DV}, which is also known as \emph{4-trampoline}, is a directed path graph but not an RDV graph since it is not even strongly chordal \cite{Farber83}.
For any choice of root, every path $T(v)$ becomes a subtree with at most two leaves, and so the greedy-algorithm can be implemented in $O(n\log n)$ time (presuming the clique-tree was small).   Unfortunately, this does not necessarily give a maximum matching, see Figure~\ref{fig:DV}.

\begin{figure}[ht]\centering
\includegraphics[width=0.99\linewidth,page=1]{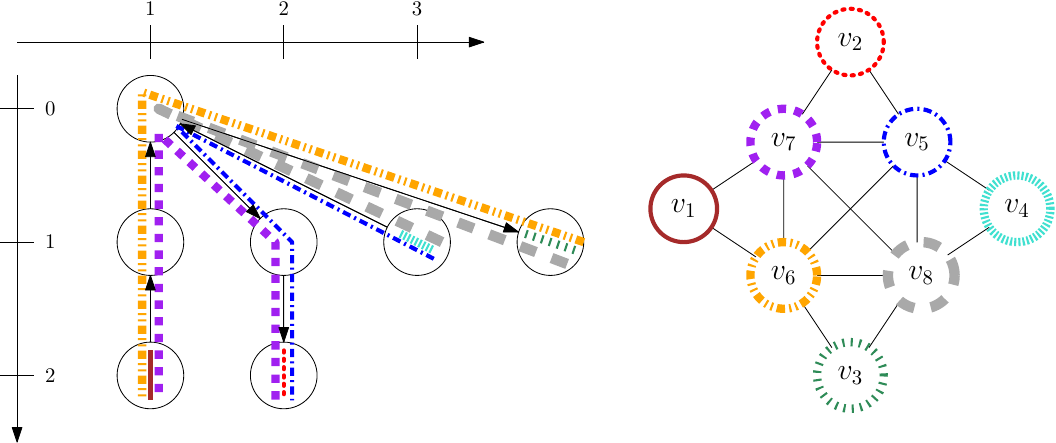}
\caption{A directed path graph that is not strongly chordal. With the depicted bottom-up enumeration, the greedy-algorithm would choose matching $(v_5,v_2)$, $(v_6,v_1)$, $(v_8,v_3)$ and leave $v_4,v_7$ unmatched even though the graph has a matching of size 4.}
\label{fig:DV}
\end{figure}

This raises another natural open problem:   Can we find a maximum matching in a directed path graph (with a given small clique-tree) in $O(n\log n)$ time?   
How about the path graphs, an even broader class?


\small
\bibliographystyle{abbrv} 
\bibliography{ref} 

\newpage
\section*{Appendix}
\section{Proof of Theorem~\ref{thm:RDVstronglyChordal}}

To prove the theorem, we first need some definitions.   Write $N[v]$ for the \emph{closed neighbourhood} of a vertex, i.e. the set consisting of  $v$ and all its neighbours.   Call a vertex $v$ \emph{simple} \cite{Farber83} if
$N[v]$ is a clique 
that can be ordered as $w_1,\dots,w_k$ such that $N[w_1]\subseteq N[w_2]\subseteq \dots \subseteq N[w_k]$.   The crucial ingredient is the following observation:

\begin{lemma}
Let $G$ be a graph with an RDV representation. Then a vertex $v_1$  that maximizes $y(t(v_1))$ is simple.
\end{lemma}
\begin{proof}
Since $v_1$ maximizes $y(t(v_1))$, node $t(v_1)$ must belong to $P(w)$ for any neighbour $w$ of $v$.   This shows immediately that $N[v_1]$ is a clique since all subtrees of neighbours share $t(v_1)$.   

Now remove all nodes from the RDV representation that have $y$-coordinate strictly bigger than $y(t(v_1))$; by choice of $v_1$ this does not remove any adjacencies.   
If we now sort the neighbours of $v$ as $w_1,\dots,w_k$ by decreasing $y$-coordinate of their top endpoints, then (since all subtrees are downward paths that end at $t(v_1)$) we have $P(w_1)\subseteq \dots \subseteq P(w_k)$ and so $v_1$ is simple.
\end{proof}

We can view a bottom-up elimination order $v_1,\dots,v_n$ as using a vertex $v$ that maximizes $y(t(v))$ as $v_1$, removing it from the graph, and repeating until the graph is empty.    By the above, then each $v_i$ is a simple vertex with respect to the graph induced by $\{v_i,\dots,v_n\}$. Farber \cite[Theorem 3.3]{Farber83} showed that a vertex order with this property is a strong elimination order, and as mentioned earlier, using a strong elimination order guarantees that the greedy-algorithm for matching succeeds \cite{DAHLHAUS199879}.

\section{The graph of Figure~\ref{fig:RDV_example}}

We claimed earlier that the graph $G$ in Figure~\ref{fig:RDV_example} is neither a circular arc graph nor a permutation graph, and we briefly argue this here.
We repeat the graph here for convenience.

\begin{figure}[ht]
\includegraphics[width=0.65\linewidth,page=3]{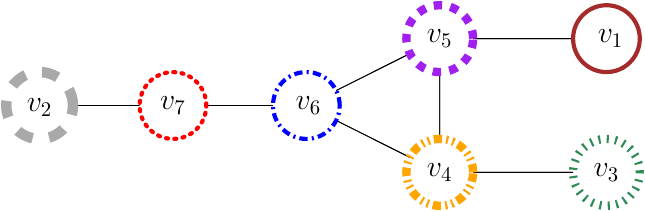}\hspace*{\fill}
\includegraphics[width=0.25\linewidth,page=4]{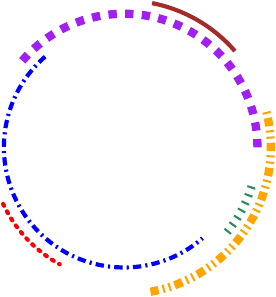}
\caption{The graph $G$ of Figure~\ref{fig:RDV_example} and a circular arc representation of $G\setminus \{v_2\}$.}
\end{figure}

Most of our argument considers only the graph $G\setminus \{v_2\}$.   This is well-known not to be a comparability graph \cite[Figure 5.1]{Gol80},
and since permutation graphs are subgraphs of comparability graphs and closed under vertex-deletion, $G$ is not a permutation graph.

Next observe that  vertices
$\{v_1,v_3,v_7\}$ form what is known as an \emph{asteroidal triple}: any two of them can be connected via a path that avoids the neighbourhood of the third.   
No such structure can exist in an interval graph.    
In fact, $G\setminus \{v_2\}$ is known to be an obstruction for \emph{Helly circular-arc graphs} \cite{LinS09}, i.e., it does not have a circular arc representation where for every clique $C$ the arcs of vertices in $C$ all share a common point.   Since $\{v_4,v_5,v_6\}$ is the only non-trivial clique, therefore in any circular arc representation of $G\setminus \{v_2\}$ the three arcs of $v_4,v_5,v_6$ do not share a common point.   To still have pairwise intersections, these three arcs together cover the entirety of the circle.    But then we cannot add an arc for $v_2$ anywhere since it has no edge to any of these three vertices.    Therefore $G$ is not a circular arc graph.

\end{document}